\documentclass[11pt]{article}

\usepackage{etoolbox}
\usepackage[whole]{bxcjkjatype}
\usepackage[linesnumbered,ruled,vlined]{algorithm2e}
\SetKw{Break}{break}
\usepackage{amsmath,amssymb,amsthm,mathtools}
\usepackage{setspace}
\usepackage{mathrsfs}

\usepackage{amsmath,amssymb,amsthm}
\usepackage{algorithmic}

\newcommand{\bix}{\boldsymbol{x}}

\newcommand{\biy}{\boldsymbol{y}}

\usepackage{authblk}
\usepackage{array,multirow,graphicx}
\usepackage{comment}
\usepackage[svgnames]{xcolor}
\usepackage{tikz}
\usepackage{pgfplots}
\usepackage{pgfplotstable}
\usetikzlibrary{positioning,intersections,calc}
\pgfplotsset{compat=newest}
\usepackage{url}
\usepackage[noadjust]{cite}
\usepackage[colorinlistoftodos]{todonotes}
\usepackage{bm}
\renewcommand{\mid}{\,:\,}

\theoremstyle{plain}
\newtheorem{theorem}     {Theorem}
\newtheorem{lemma}       {Lemma}

\newtheorem{definition}  {Definition}

\makeatletter
\newcommand{\subalign}[1]{%
  \vcenter{%
    \Let@ \restore@math@cr \default@tag
    \baselineskip\fontdimen10 \scriptfont\tw@
    \advance\baselineskip\fontdimen12 \scriptfont\tw@
    \lineskip\thr@@\fontdimen8 \scriptfont\thr@@
    \lineskiplimit\lineskip
    \ialign{\hfil$\m@th\scriptstyle##$&$\m@th\scriptstyle{}##$\crcr
      #1\crcr
    }%
  }
}
\makeatother

\usepackage[top=30truemm,bottom=30truemm,left=25truemm,right=25truemm]{geometry}

\usepackage{setspace}
\begin{document}

\begin{center}
{\textbf {\LARGE A constant-ratio approximation algorithm for a class of hub-and-spoke network design problems and metric labeling problems: star metric case}}\\
\vspace{5mm}
\begin{tabular}{ccc}
 {\large Yuko Kuroki and Tomomi Matsui}\\
 {\small Department of Industrial Engineering and Economics, Tokyo Institute of Technology}\\
  \texttt{kuroki.y.aa@m.titech.ac.jp, matsui.t.af@m.titech.ac.jp}
\end{tabular}
\end{center}



\pagestyle{plain}
\pagenumbering{roman}

\begin{abstract}
Transportation networks frequently employ hub-and-spoke network architectures to route flows between many origin and destination pairs.
In this paper,
we deal with a problem,
called the 
	{\em single allocation hub-and-spoke network design problem}.
In the single allocation hub-and-spoke network design problem,
the goal is to allocate each non-hub node to exactly one of given hub nodes so as to minimize the total transportation cost.    
The problem is essentially equivalent to another
combinatorial optimization problem, called the {\em metric labeling} problem. 
The metric labeling problem was first introduced by Kleinberg and Tardos~\cite{KLEINBERG2002} in 2002,
motivated by application to segmentation problems in computer vision and energy minimization problems in related areas.

In  this paper,
we deal with the case where the set of hubs forms a star, which is called the {\em star-star hub-and-spoke network design problem}, and the  {\em star-metric labeling problem}.
This model arises especially in telecommunication networks in the case where set-up costs of hub links 
	are considerably large
	or full interconnection is not required.
We propose a polynomial-time randomized approximation algorithm for these problems, whose approximation ratio is less than 5.281.
Our algorithms solve a linear relaxation problem and apply dependent rounding procedures. 


\end{abstract}

\setstretch{1.20}

\pagenumbering{arabic}
\setcounter{section}{0}
\section{Introduction}\label{se:intro}

Design of efficient networks is desired  in transportation systems, such as telecommunications, delivery services, and airline operations,
and is one of the extensively studied topics in operations research field.
 Transportation networks frequently employ {\em hub-and-spoke network} architectures to route flows between many origin and destination pairs.
A transportation network with many origins and destinations requires a huge cost,
and hub-and-spoke networks play an important role in reducing transportation costs and set-up costs.
Hub facilities work as switching points for flows in a large network.
Each non-hub node is allocated to exactly one of the hubs instead of assigning every origin-destination pair directly.
Using hub-and-spoke architecture,
we can construct large transportation networks with fewer links, which leads to smart operating systems (see Figure~\ref{fig:pointhub}).
\vspace{1.0cm}

\begin{figure}[htpp]
\begin{center}
			\includegraphics[width=9.5cm]{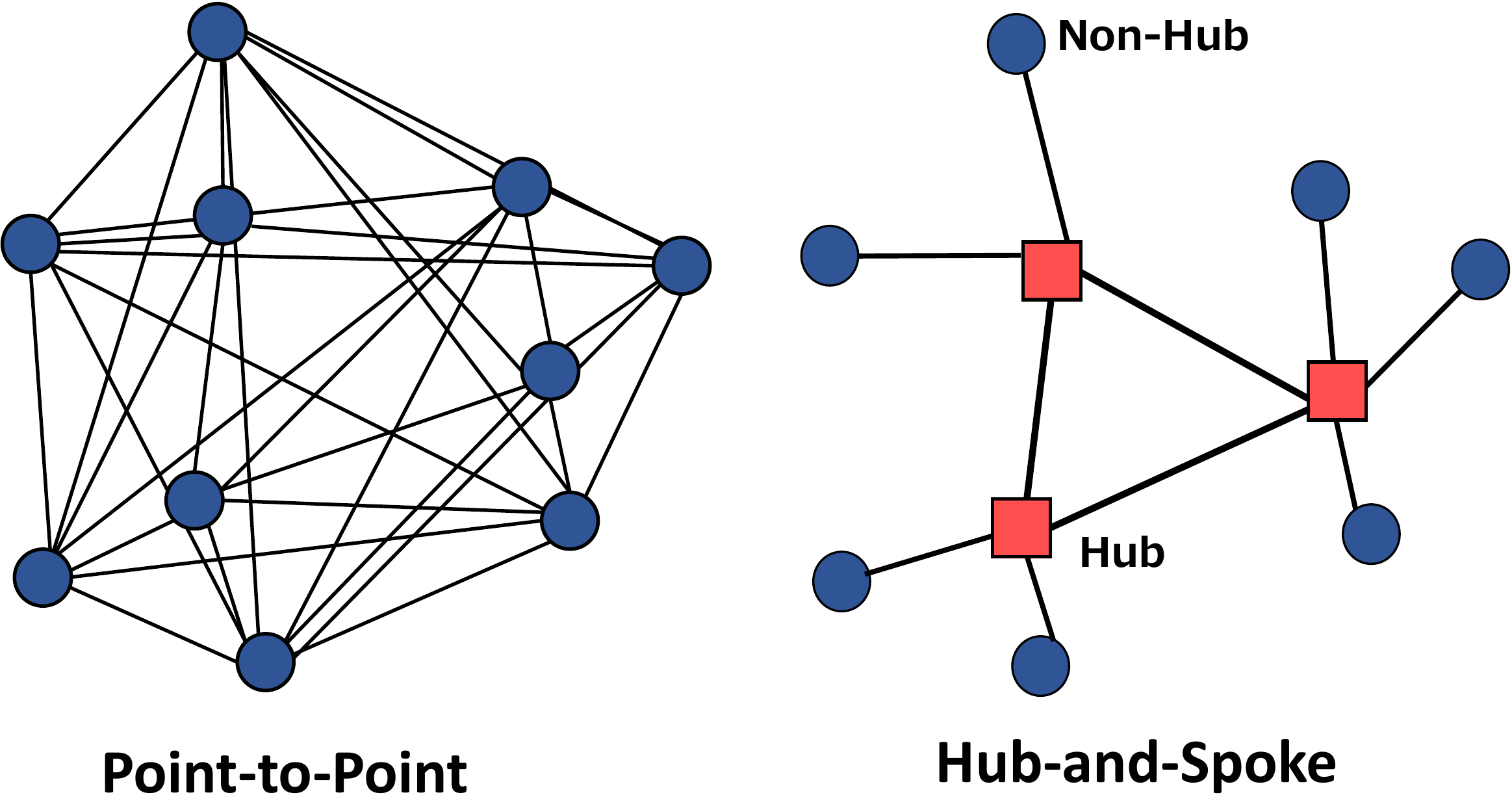}
			\caption{Point-to-Point vs. Hub-and-Spoke}
			\label{fig:pointhub}
\end{center}
\end{figure}



\subsection{Single Allocation Hub-and-Spoke Network Desing Problem}

In real transportation systems, the location of hub facilities is often fixed because of costs for moving equipment on hubs. In that case, the decision of allocating non-hubs to hubs is much important for an efficient transportation.
In this study, we discuss the situation where  
	the location of the hubs is given,
	and deal with a problem,
	called a 
	{\em single allocation hub-and-spoke network design problem}, which aims to minimize the total transportation cost.

Formally, the input consists of an $h$-set $H$ of hubs, an $n$-set $N$ of non-hubs,
 non-negative cost per unit flow $c(i,j)=c(j,i)$ for each pair $\{ i,j\} \in H^2$,
 and $c(p,i)$ for each ordered pair $(p,i) \in N \times H$.
 Additionally, we are given $w(p,q)$ which denotes a non-negative amount of flow from non-hub $p$ to another non-hub $q$.
 The task is to find an assignment $f: N \rightarrow H$, that maps non-hubs to hubs minimizing the total transportation cost $Q(f)$ defined below.
 The transportation cost corresponding to a flow from non-hub $p$ to non-hub $q$ is defined by $w_{pq}(c(p,f(p))+c(f(p),f(q))+c(f(q),q))$.
 Thus 
 \[
Q(f)=\sum_{(p,q) \in N^2} w(p,q)\left(c(p,f(p))+c(q,f(q))+c(f(p),f(q)) \right)
 ,\]
 and the goal is to find an assignment that minimizes the total transportation cost.

    When the number of hubs is equal to two,  there exist polynomial time exact algorithm~\cite{greig1989exact,SOHN1997}. 
    Sohn and Park~\cite{SOHN2000} proved NP-completeness 
	of the problem even if the number of hubs is equal to three.
In the case where the given matrix of costs between hubs is a {\em Monge matrix},
there  exists a polynomial-time exact algorithm~\cite{chekuri2004linear}.
Iwasa et al.\@~\cite{IWASA2009}
	proposed a simple deterministic $3$-approximation algorithm 
	and a randomized $2$-approximation algorithm
	under the assumptions that $c_{ij} \leq c_{pi} + c_{pj}$  $(\forall (i,j,p)\in H^2\times N)$ and $c_{ij} \leq c_{ik}+c_{kj}$  $(\forall (i,j,k)\in H^3)$.
They also proposed a $(5/4)$-approximation algorithm 
	for the special case where the number of hubs is three.
Ando and Matsui~\cite{ANDO2011} deal with the case in which
	all the nodes are embedded in a 2-dimensional plane 
	and the transportation cost of an edge per unit flow is proportional to the Euclidean distance 
	between the end nodes of the edge.
They proposed a randomized $(1+2/\pi)$-approximation algorithm. In the previous our work~\cite{kuroki2017}, we proposed $2(1+\frac{1}{h})$-approximation algorithm for the case where the set of hubs forms a cycle.




 \subsection{Metric Labeling Problem}
In 2002, Kleinberg and Tardos~\cite{KLEINBERG2002} introduced the {\em metric labeling} problem, motivated by applications to segmentation problems in computer vision and energy minimization problems in related areas. A variety of heuristics that use classical combinatorial optimization techniques have developed in these fields (\cite{boykov1998markov,boykov2001fast,kumar2014rounding,li2016complexity} for example).
 A single allocation hub-and-spoke network design problem
	includes a class of the metric labeling problem.
 The metric labeling problem captures a broad range of classification problems and has connections to Markov random field.
In such classification problems, the goal is to assign labels to some given set of objects minimizing the total cost of labeling.

Formally, the metric labeling problem takes as input an $n$-vertex undirected graph $G(V,E)$ with a nonnegative weight function $w$ on the edges, a set $L$ of labels with metric distance function $d: L \times L \rightarrow R $ associated with them,
and an assignment cost $c(v,a)$ for each vertex $v \in V$ and label $a \in L$. The output is an assignment for every object $v \in V$ to a label $a \in L$. 
Given a solution $f:V \rightarrow L$ to the metric labeling, the quality of labeling $Q(f)$ is based on the contribution of two sets of terms.

{\em Vertex labeling cost}:
For each object $v \in V$, this cost is denoted by $c(v,f(v))$. A vertex labeling cost $c(v,a)$ express an estimate of its likelihood of having each label $a \in L$. 
These likelihoods are observed from some heuristic preprocessing of the data.
For example, suppose the observed color of pixel (i.e., object) $v$ is white; then the cost $c(v,black)$ should be high while $c(v,white)$ should be low.

{\em Edge separation cost}:
For each edge $e=\{u,v\} \in E$, the cost is denoted by $w(\{u,v\}) \cdot d(f(u),f(v))$.
The weights of the edges express a prior estimate on relationships among objects; if $u$ and $v$ are deemed to be related, then we would like them to be assigned close or identical labels.
A distance $d(a,b)$ for $a,b \in L$ represents  how similar label $a$ and $b$ are.
For example, $d(white,black)$ would be large while $d(orange,yellow)$ would be small.
If we assign label $a$ to object $u$ and label $b$ to object $v$, then we pay $w(\{u,v\})d(a,b)$ as the edge separation cost.

\noindent Thus,
\[
Q(f)=\sum_{u \in V}c(u,f(u))+\sum_{\{u,v\}\in E}w(\{u,v\})d(f(u),f(v))
,\]
and the goal is to find a labeling $f: V \rightarrow L$ minimizing $Q(f)$.
Due to the simple structure and variety of applications,
the metric labeling has received much attention since its introduction by Kleinberg and Tardos~\cite{KLEINBERG2002}. 

In case the number of labels is two, the problem can be solved precisely in polynomial-time.
The first approximation algorithm for the metric labeling problem was shown by Kleinberg and Tardos~\cite{KLEINBERG2002}, and its approximation ratio is O$(\log k \log\log k)$, where $k$ denotes the number of labels.
This algorithm uses the probabilistic tree embedding tequnique~\cite{bartal1998approximating}.
Using the improved representation of metrics as combination of tree metrics by Fakcharoenphol, Rao, and Talwar~\cite{Fak2004}, its approximation ratio was improved to O$(\log k)$, which is the best general result to date.
Constant-ratio approximations are known for some special cases~\cite{archer2004approximate,chekuri2004linear,calinescu2005approximation,KLEINBERG2002}.


 \subsection{Contributions}\label{sub:contribution}


 We deal with the a single assignment hub-and-spoke network design problem where the given set of hubs forms a star,
 and corresponding problem is called the 
 {\em star-star hub-and-spoke network 
 design problems} and {\em star-metric labeling  problems}.
 In this case,
 each hub is only connected to a unique depot.
 When all the transportation cost per unit flow between the depot and each hub are same,
 this problem is equivalent to the {\em uniform labeling}
 problem  (all distances of labels are equal to 1) introduced in \cite{KLEINBERG2002} which is still NP-hard.
 For star-metric case,
 using the result of \cite{konjevod2001approximating} for planer graphs,
 there exists O$(1)$-approximation algorithm~\cite{KLEINBERG2002}. 
 The previous O$(1)$-approximation ratio is at least 6.
 We proposed an improved approximation algorithm for star-metric case, and the approximation ratio is  $\min \{\frac{r-1}{\log r}\left( 2+ \frac{r^2+1}{r^2-1}\right) | r>1 \}  ( \approx 5.2809 \ {\rm at}\ r \approx 1.91065)$.
Our results give an important class
	of the metric labeling problem and hub-and-spoke network design problems,
	which has a polynomial time approximation algorithm 
	with a constant approximation ratio. 
In case where set-up costs of hub links 
	are considerably large, incomplete networks can be used instead of full interconnection among hub facilities.
The star structures, that we discuss in this paper, frequently arise in  especially telecommunication networks~\cite{LABBE2008}.

\subsection{Related Work}\label{sec:related}

\noindent\textbf{Approximation Results for Metric Labeling Problems.}
 Gupta and Tardos~\cite{gupta2000constant} considered an important case of the metric labeling problem,
 in which the metric is the {\em truncated linear metric} where the distance between $i$ and $j$ is given by $d(i,j)= \min\{M,|i-j|\}$.
Chekuri et al.~\cite{chekuri2004linear} proposed $(2+\sqrt{2})$-approximation algorithm for the truncated linear metric, which is best known result.

In the case where the metric $d$ on a set of labels $L$ is a planar metric,
there exists O$(\log \mathrm{diam} \ G')$-approximation to the problem from the result~\cite{konjevod2001approximating} and \cite{KLEINBERG2002},
where $G'=(L,E,w)$ denote the weighted connected graph.
 Konjevod et al.~\cite{konjevod2001approximating} showed that
 for any positive integer $s$,
 the metric of $G$  without a $K_{s,s}$ minor can be probabilistically approximated by a special case of tree metric,
called {\em $r$-hierarchically well separated tree {\rm(}r-HST{\rm)}} with distortion O$(\log \mathrm{diam} \ G)$.
 Kleinberg and Tardos~\cite{KLEINBERG2002} gave a constant ratio approximation algorithm to the metric labeling for the case where the metric $d$ on a set of labels is the $r$-HST metric. Then O$(\log \mathrm{diam} \ G')$-approximation was guaranteed by combining these results for this case.

  \hspace{2cm}
\begin{table}[htb]
\begin{center}
\caption{Existing approximation algorithms for metric labeling problems}
\vspace{0.5cm}
\begin{tabular}{|c|c|}
\hline
Metric &App. Ratio\\
\hline
 general & O$(\log k)$~\cite{KLEINBERG2002,Fak2004} \\
\hline
 planar graph & O$(\log \mathrm{diam} \ G') $~\cite{konjevod2001approximating,KLEINBERG2002} \\
\hline
 truncked linear  & $2+\sqrt{2}$~\cite{chekuri2004linear}\\
\hline
 uniform & 2~\cite{KLEINBERG2002} \\
\hline
\end{tabular}
\end{center}
\end{table}

\noindent\textbf{Inapproximability Results.}
Chuzhoy and Naor~\cite{chuzhoy2007hardness} showed that there is no polynomial time approximation algorithm
	with a constant ratio for the metric labeling problem unless $\mbox{P}=\mbox{NP}$.
Moreover, they proved that the problem is $\Omega((\log |V|)^{1/2-\delta})$-hard to approximate for any constant $\delta satisfying 0<\delta<1/2$, unless NP$\subseteq$DTIME$(n^
{poly(\log n)})$ (i.e. unless NP has quasi-polynomial time algorithms).    

In 2011, Andrew et al.~\cite{andrews2011capacitated}
introduced {\em capacitated metric labeling}, in which there are additional restrictions that each label $i$ 
receives at most $l_i$ nodes.
They proposed a polynomial-time, O$(\log |V|)$-approximation algorithm when the number of labels is fixed and proved that it is impossible to approximate the value of an instance of capacitated metric labeling  to within any finite ratio, unless P = NP.\\
    

\noindent\textbf{Hub Location Problems.}
{\em Hub location problems} ({\em HLPs}) consist of locating hubs 
	and designing hub networks so as to minimize the sum of  set-up costs and transportation costs.
HLPs are formulated as a quadratic integer programming problem by O'Kelly~\cite{OKELLY1987} in 1987.
Since O'Kelly proposed HLPs,
hub location has been studied by researchers in different areas such as location science, geography, operations research,
regional science, network optimization,
transportation, telecommunications, and computer science.
Many researches on HLPs 
	have been done in various applications and there exists several reviews and surveys
(see~\cite{ALUMUR2008,CAMPBELL1994,%
CAMPBELL2012,CONTRERAS2015,%
KLINCEWICZ1998,OKELLY1994} for example).
In case where the location of the hubs is given,
the remaining subproblem is essentially equal to the {\em  single allocation hub-and-spoke network design problem} mentioned in previous subsections.
Fundamental HLPs assume a full interconnection between hubs.
Recently, several researches consider 
	incomplete hub networks which arise especially 
	in telecommunication systems
	(see~\cite{CAMPBELL2005a,CAMPBELL2005b,ALUMUR2008,CALIK2009}
	for example). 
These models are useful when set-up costs of hub links 
	are considerably large
	or full interconnection is not required.
	That motivated us to consider a single allocation hub-and-spoke network design problem
	where the given set of hubs forms a star (see Figure~\ref{FIG:topology}).
There are researches 
	which assume that hub networks constitute 
	a particular structure such as 
a line~\cite{MARTINS2015},
a cycle~\cite{cycle_contreras2017exact},
tree~\cite{tree_tavakkoli2014multi,tree_sedehzadeh2016,KIM1992,CONTRERAS2009,CONTRERAS2010,DESA2013},
a star~\cite{LABBE2008,YAMAN2008,YAMAN2012}.

\begin{figure}[htpp]
\begin{center}
			\includegraphics[width=10cm]{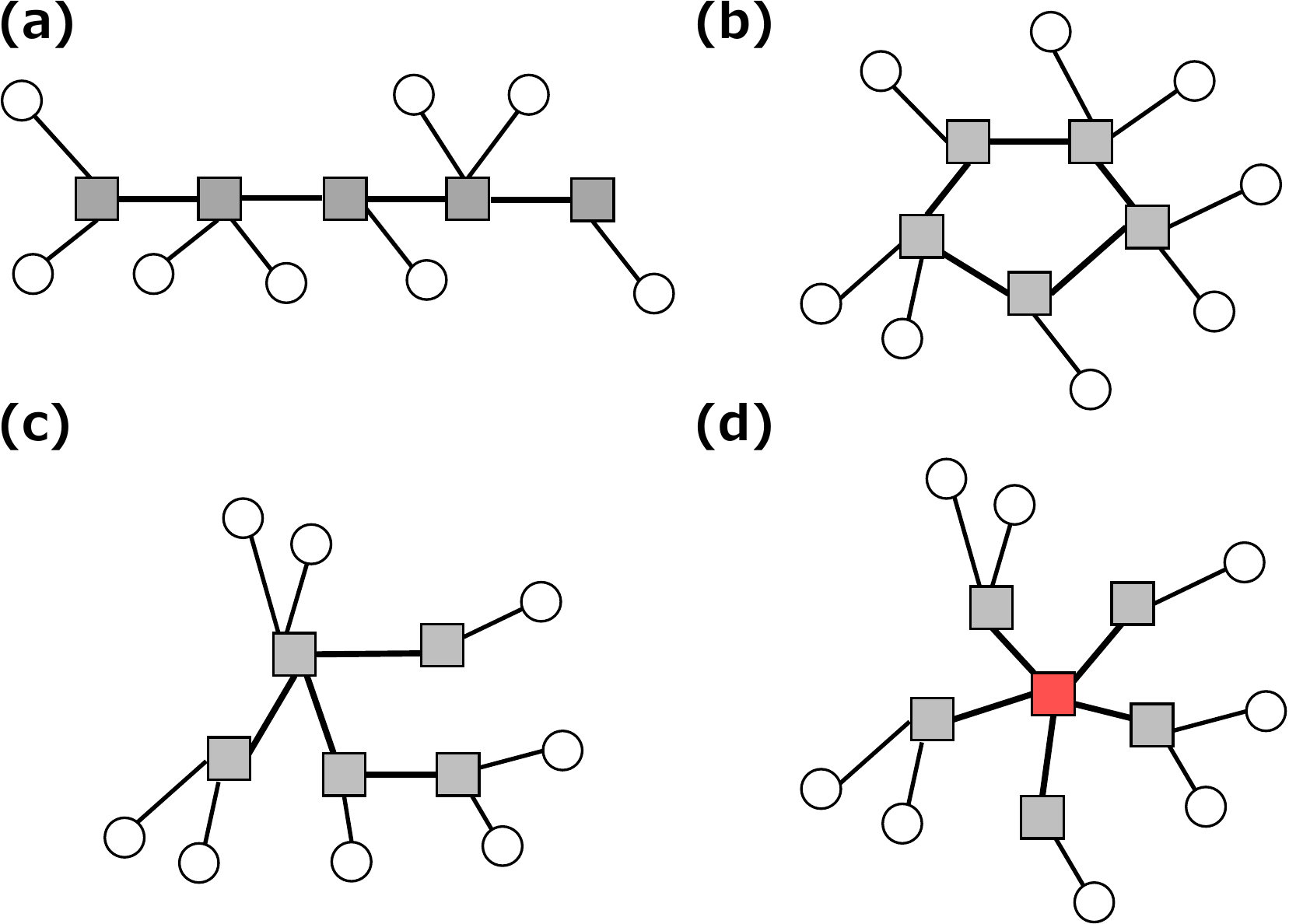}
			\caption{Structure of (a) line-star,
            (b) cycle-star,
            (c) tree-star,
            and (d) star-star}
			\label{FIG:topology}
\end{center}
\end{figure}

\subsection{Paper Organization}

This paper is structured as follows
In Section~\ref{sec:starprob},
we provide a problem formulation.
In Section~\ref{sec:algstar},
we describe an approximation algorithm.
In Section~\ref{sec:analysis}, we analyze the approximation ratio of our algorithm.



\section{Problem Formulation}\label{sec:starprob}
Let $H=\{1,2,\ldots, h\}$ be a $h \ (\geq 3)$-set of hub nodes and 
	let $N=\{p_1,p_2,\ldots,p_n \}$ be a $n$-set of \mbox{non-hub} nodes.
This paper deals with a single assignment hub network design problem
	which assigns each non-hub node to exactly one hub node.
We discuss the case in which the set of hubs forms a star, 
	and the corresponding problem is called 
	the {\em star-star hub-and-spoke network design problem} and/or {\em star-metric labeling problem}.
More precisely, 
	we are given a unique depot, denoted by $0$, which lies at the center of hubs.
    Each hub $i \in H$ connects to the depot and doesn't connect to other hubs.
    Let $\ell_i$ be the transportation cost per unit flow between the depot and a hub $i$.
    In our setting, we assume that $0 \leq \ell_1 \leq \ell_2 \leq \cdots \leq \ell_h$ and $ \ell_i \in  \mathbb{Z}$  for all $i \in H$.
    Then for each pair of hub nodes $(i,j) \in H$,
    $c_{ij}$ denotes the transportation cost per unit flow between hub $i$ and hub $j$ and it satisfies that $c_{ij}=\ell_i+\ell_j$.
    We assume $c_{ii}=0$ for all $i \in H$.
For each ordered pair $(p ,i) \in N \times H$, 
	$c_{pi}$ denotes a non-negative cost per unit flow 
	on an undirected edge $\{p,i\}$.
We denote a given non-negative amount of flow from a non-hub $p$ 
	to another non-hub $q$ by $w_{pq}\ (\geq 0)$.
Throughout this paper, we assume that $w_{pp}=0 \; (\forall p \in N)$.
We discuss the problem for finding an assignment of non-hubs to hubs 
	which minimizes the total transportation cost defined below.
    
\begin{figure}[htpp]
\begin{center}
			\includegraphics[width=13cm]{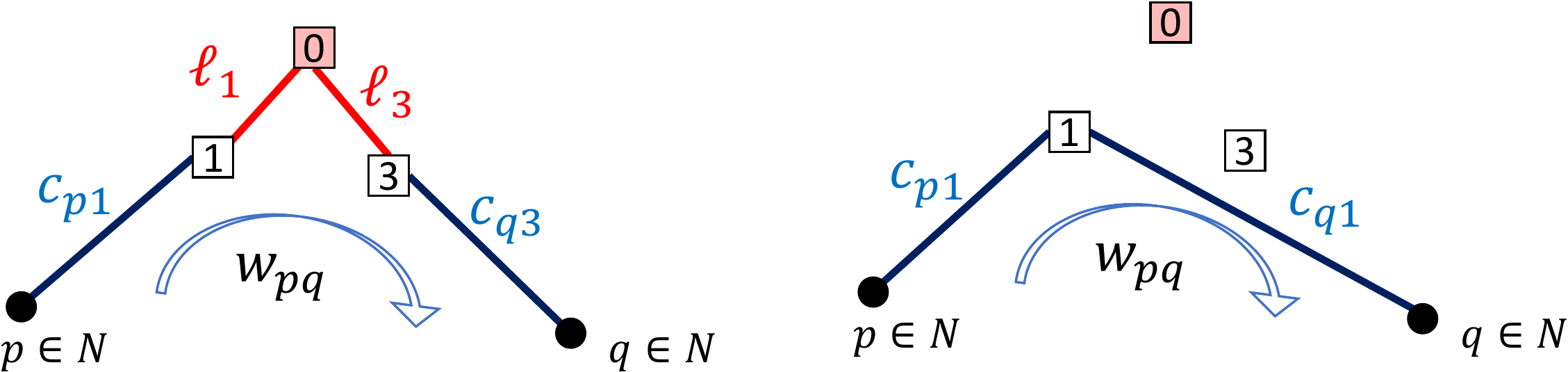}
			\caption{Examples of transportation}
			\label{fig:example}
\end{center}
\end{figure}

When non-hub $p$ and non-hub $q$ $(p \neq q)$
	are assigned to hub $i$ and hub $j$, respectively,
	an amount of flow $w_{pq}$ is sent along a path
	$((p,i), (i,0), (0,j), (j,q))$.
In the rest of this paper,
	 a matrix $C=(c_{ij})$ defined above 
	is called the {\em cost matrix}
	and/or the {\em star-metric matrix}.
The transportation cost corresponding to a flow 
	from the origin $p \in N$ to destination $q \in N$ is defined by
	$w_{pq}(c_{pi}+c_{ij}+c_{qj})$.
    
In case where $h=3$ and $\ell_1=\ell_2=\ell_3=1$,
the corresponding problem is equivalent to the problem 
where a $3$-set of hubs forms a complete graph and $C$ satisfies that $c_{12}=c_{23}=c_{31}=2$.
Thus the star-star hub network design problem
	is NP-hard~\cite{SOHN2000}.\\

Now we formulate our problem as 0-1 integer programming.
First, we introduce a $0$-$1$ variable $x_{pi}$ 
	for each pair  $\{p,i\} \in N \times H$ as follows:
\[	x_{pi}=	\left\{
				\begin{array}{ll}
					1 & (\mbox{$p \in N$ is assigned to $i \in H$}), \\
					0 & (\mbox{otherwise}). 
				\end{array}
			\right.
\]
\noindent
Since each non-hub is connected to exactly one hub,
	we have a constraint $\sum_{i \in H}x_{pi}=1$ for each $p \in N$.
Then, the star-star hub network design problem (star-metric labeling problem) can be formulated as follows:
\begin{alignat*}{6}
&\mbox{\rm SHP: \quad} && \mbox{\rm min.\quad }
&& 
		\sum_{(p,q) \in N^2,\ p \not=q} w_{pq}
		\left( 
			  \sum_{i\in H} c_{pi}x_{pi}   
			\right.&&\left.+ \sum_{j\in H} c_{qj}x_{qj} 
		+
			\sum_{k \in H}\ell_{k}|x_{pk}-x_{qk}| 
		\right)  
	\\
&&& \mbox{\rm s.~t.}		
&& \sum_{i\in H} x_{pi} = 1  && (\forall p \in N), \\ 
&&&&&  x_{pi} \in \{0,1\}		&&(\forall \{p,i\} \in N \times H).\\
\end{alignat*}
%


 Next we describe a linear relaxation problem.
 By substituting  non-negativity constraints of the variables $x_{pi} \ (\forall  \{p,i\}  \in N \times H)$ for $0$-$1$ constraints in SHP and replace $|x_{pk}-x_{qk}|$ with $Z_{pqk}$, we obtain the following a linear relaxation problem  denoted by LRP.
\begin{alignat*}{6}
&\mbox{\rm LRP: \quad} && \mbox{\rm min.\quad }
&& 
		\sum_{(p,q) \in N^2,\ p \not=q} w_{pq}
		\left( 
			  \sum_{i\in H} c_{pi}x_{pi}   
			\right.&&\left.+ \sum_{j\in H} c_{qj}x_{qj} 
		+
			\sum_{k \in H}\ell_{k}Z_{pqk}
		\right)  
	\\
&&& \mbox{\rm s.~t.}		
&& \sum_{i\in H} x_{pi} = 1  && (\forall p \in N), \\ 
&&&&&  0\leq x_{pi}	&&(\forall  \{p,i\} \in N \times H), \\
&&&&& -Z_{pqk}\leq x_{pk}-x_{qk} \leq Z_{pqk} 	&& \ (\forall (p,q) \in N^2, \forall k \in H).
\end{alignat*}
We can solve LRP in polynomial time by employing 
	an interior point algorithm.

\section{Algorithm}\label{sec:algstar}
We now design an approximation algorithm. 
The approach is proceeded as follows:
\begin{description}
	\item[Step~1.] Choose $\lambda \in [0,1)$ uniformly at random and classify the hubs under $\kappa_{\max} +1$ classes according to {\bf Definition~\ref{definitionA}}.
    \item[Step~2.] Solve the linear relaxation problem LRP and obtain an optimal solution $\bix^*$.
	\item[Step~3.] Find a partition of non-hubs by {\bf Algorithm~\ref{alg:class}}.
	\item[Step~4.] Assign each non-hub to a hub by {\bf Algorithm~\ref{alg:assign}}.
\end{description}

Now, we describe our algorithm precisely.
 In Step 1, we classify the set of hubs according to the distance between each hub and the depot (see Figure~\ref{fig:hubclass}). We assign each hub to a class.
This classification is based on the following definition.

\begin{definition}\label{definitionA}
For any $\lambda \in {[0,1)}$, we say that hub $i$ belongs to class $\kappa$
	if and only if $\ell_i \ (\geq 1)$ satisfies the inequality $r^{\max \{(\kappa-2)+\lambda, 0 \}} \leq \ell_i < r^{(\kappa-1)+\lambda} $ and hub $i$ belongs to class $0$ if and only if $\ell_i=0$, where $\kappa$ is a non-negative integer.
\end{definition}

\begin{figure}[htpp]
\begin{center}
			\includegraphics[width=12cm]{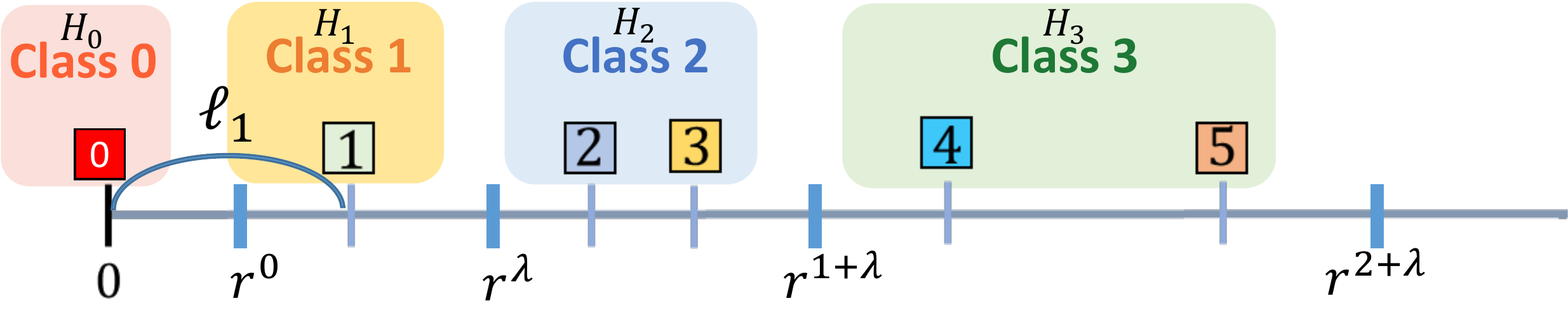}
			\caption{Classification of Hubs}
			\label{fig:hubclass}
\end{center}
\end{figure}
Before we describe the details of later steps,
we introduce some notations.
Let $\alpha(\lambda,i)$ be the class of hub $i \in H$.
We denote a subset of integers $\{0,1,2,...,\kappa_{{\rm max}} \}$ by $ [\kappa_{\rm max}]$ where $\kappa_{\rm max}= \max_{i \in H}\alpha(\lambda,i)$. 
Let $H_{\kappa}$ be a subset of hubs that belongs to class $k \in [\kappa_{\max}]$ i.e. $H_\kappa=\{ i \in H \ | \ \alpha(\lambda,i)=\kappa \}$.
Let $\beta(p)$ be the class that non-hub $p \in N$ belongs to.
We denote a subset of non-hubs that belong to class $\kappa \in [\kappa_{\max}]$ by $N_\kappa$ i.e. $N_\kappa=\{ p \in N \ | \ \beta(p)=\kappa \}$.

In Step 2, we solve the linear relaxation problem  LRP formulated in the previous subsection. We use an optimal solution $ \bm{x}^*$ in \textbf{Algorithm~1} and \textbf{Algorithm~2}.

In Step 3,
we assign each non-hubs to a class of hubs defined by {\bf Definition~3}.
For example, if non-hub $p_1$ belongs to the subset $N_3$ obtained by \textbf{Algorithm~1},
$p_1$ will be assigned to one of hubs in $H_3$ defined by \textbf{Definition~3}.
Given an optimal solution of LRP and a total order of the hubs,
\textbf{Algorithm~1} outputs a partition of non-hubs, $N_0, N_1, \cdots, N_{\kappa_{\max}}$.

Here, a total order of hubs depends on labels of classes.
For example,  the total order of hubs in Figure~\ref{fig:roundingclass} is $(5, 4, 1, 3, 2)$.
The order of class labels $\pi'$  is $(\kappa_{\max}, ,\ldots,4, 2, 0, 1, 3, \ldots, \kappa_{\max}-1)$ when $\kappa_{\max} $ is an even number, and
$(\kappa_{\max}-1, \ldots,4, 2, 0, 1, 3, \ldots, \kappa_{\max})$  when $\kappa_{\max} $ is an odd number.
The order of class labels in Figure~\ref{fig:roundingclass} is $(2, 0, 1, 3)$ for example.
For each non-hub $p$, we place $x_{pi}$ for $i \in H_{\kappa}$ in the order of class labels. Then, the total order of hubs $\pi$ is defined as $\left ({\rm any \ order \ of \ hubs\ in \ } H_{\pi'(1)}, {\rm any \ order \ of \ hubs\ in \ } H_{\pi'(2)},\ldots,{\rm any \ order \ of \ hubs\ in \ }H_{\pi'(\kappa_{\max+1})} \right) $ in our rounding scheme,
where $\pi'(i)$ denotes the $i$-th element of $\pi'$.

\begin{algorithm}[t]
	\caption{Classify each non-hub into a class}\label{alg:class}
	\begin{algorithmic}[1]
		\REQUIRE{An optimal solution $\bix^*$ of LRP and a total order $\pi$ of the hubs.}
		\ENSURE{ A partition of non-hubs $N_0,N_1,\cdots,N_{\kappa_{\max}}$.}
        \STATE{Set $N_i=\emptyset \ (\forall i \in [0,1,\ldots,\kappa_{\max}]$)}
		\STATE{Generate a random variable $U$ which follows a uniform distribution defined on $[0,1)$.}
		\FOR{$p \in N$}
		\STATE{Insert non-hub $p$ into a subset $N_{\alpha(\lambda,\pi(i))}$, where $i\in \{1,2,\ldots, |H| \}$
        is the minimum number that satisfies $U < x_{p\pi(1)}+x_{p\pi(2)}+\cdots+x_{p\pi(i)}$.}
		\ENDFOR
		\STATE{\textbf{return} $N_0,N_1,\cdots,N_{\kappa_{\max}}$.} 
	\end{algorithmic}
\end{algorithm}

\begin{figure}[htpp]
\begin{center}
			\includegraphics[width=7cm]{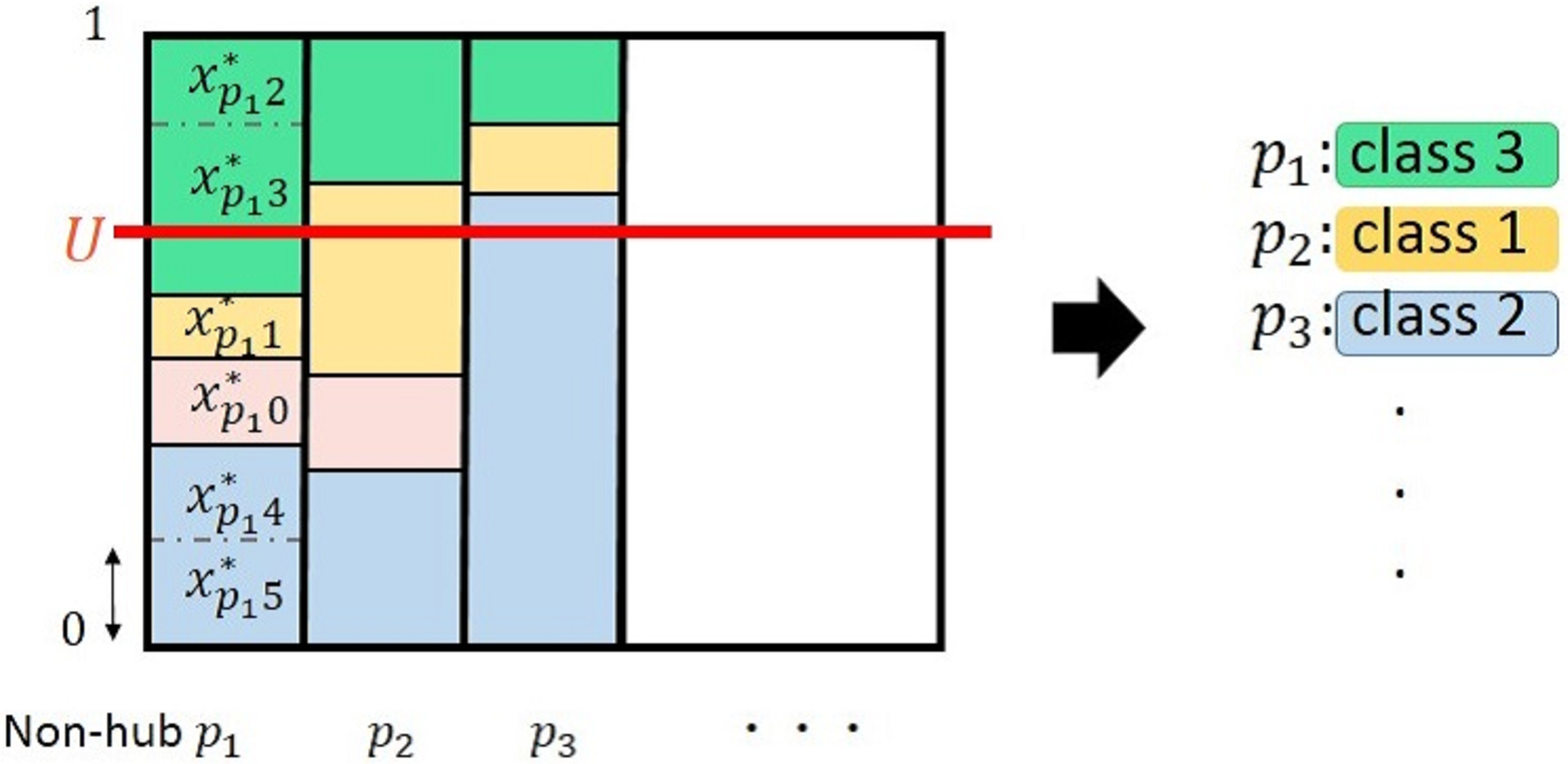}
			\caption{Dependent rounding procedure to classify each non-hub into a class}
			\label{fig:roundingclass}
\end{center}
\end{figure}

In Step 4, we decide an assignment from non-hubs to hubs using rounding technique.
In \textbf{Algorithm~2}, we perform a rounding procedure for each subset of non-hubs.
For a subset $N_{\kappa} \subseteq N$,
we first choose hub $i \in H_{\kappa}$ and $U \in[0,1)$ uniformly at random.
Then, if $U \leq x_{pi}^*$, we assign non-hub $p$ to hub $i$ (see Figure~\ref{fig:assign}).
Until all the non-hubs are assigned to one of hubs, we continue this procedure.
Note that in each phase we can set the upper bound of $U$ to the  maximum value of $x_{pi}$ of remained non-hubs.

\begin{algorithm}[t]
	\caption{Assign each non-hub to a hub}\label{alg:assign}
	\begin{algorithmic}[1]
		\REQUIRE{An optimal solution $\bix^*$ of LRP and $\kappa_{\max}+1$ subsets of non-hubs $N_0,\ldots,N_{\kappa_{\max}}$.}
		\ENSURE{An assignment from non-hubs to hubs $\boldmath{X}$.} 
		\FOR{$\kappa=0,1,\ldots,\kappa_{\max}$}
        	\STATE{Initialize $S \leftarrow N_{\kappa}$}
        	\WHILE{$|S| > 0$}
			\STATE{Choose hub $i$ $\in H_{\kappa}$ uniformly at random.}
			\STATE{Choose $U \in [0,1)$ uniformly at random.}
            \FOR{$p \in S$}
        	\STATE{\textbf{if} $U \leq x_{pi}^*$ \textbf{then} $X_{pi}=1,X_{pj}=0 \ (\forall j \in H_{\kappa} \setminus \{i\})$}
			\STATE{$S \leftarrow S \setminus \{p\}$.}
            \ENDFOR
			\ENDWHILE
        \ENDFOR
		\STATE{ \textbf{return} $\boldmath{X}$.} 
	\end{algorithmic}
\end{algorithm}

\begin{figure}[htpp]
\begin{center}
			\includegraphics[width=8cm]{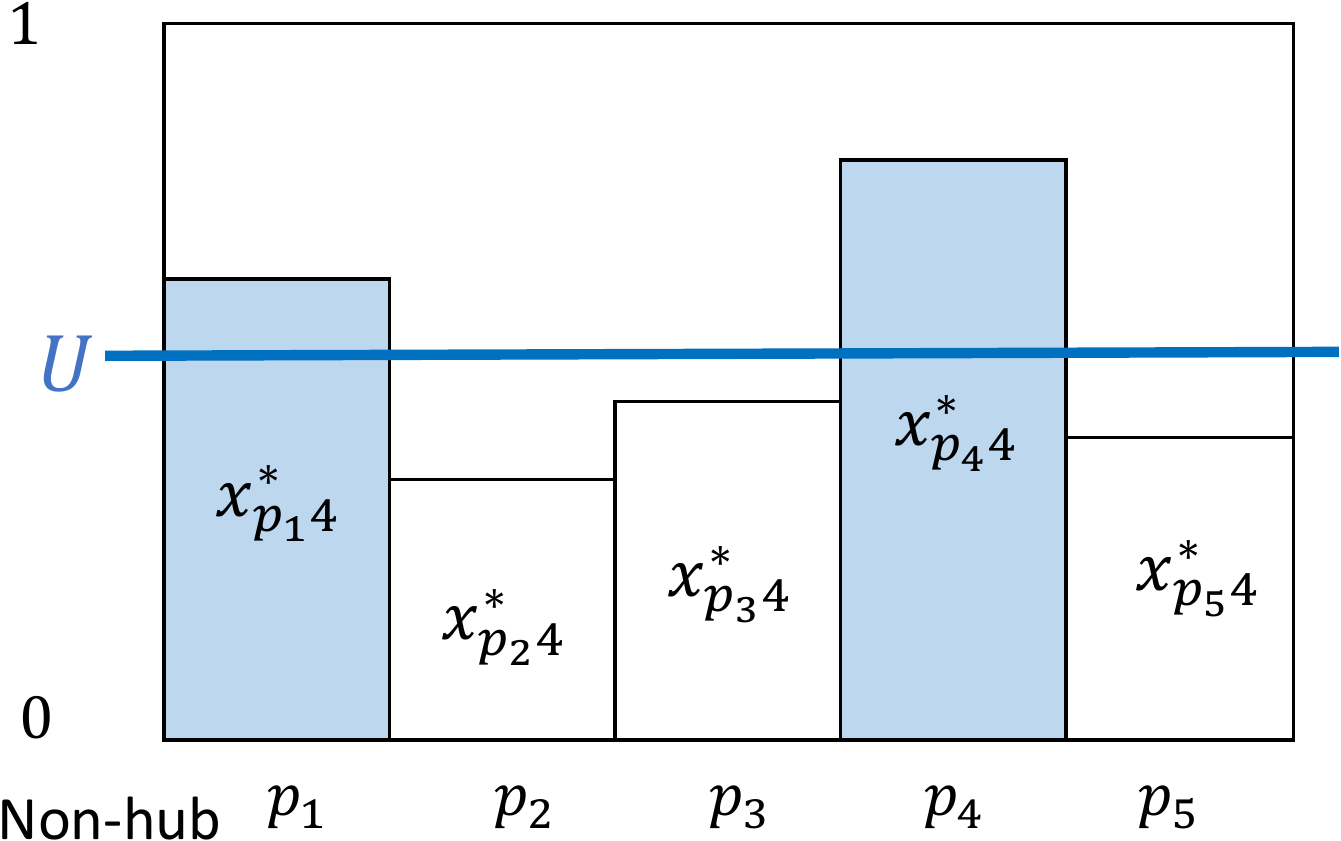}
			\caption{Non-hub $p_1$ and $p_4$ are assigned to hub $4$ by this phase, where $H_{\kappa}=\{4,5\}$ and $N_{\kappa}=\{ p_1,p_2,p_3,p_4,p_5\}$}
			\label{fig:assign}
\end{center}
\end{figure}

\section{Analysis of Approximation Ratio}\label{sec:analysis}
In this subsection, we show that our algorithm obtains a
 $\min \{\frac{r-1}{\log r}\left( 2+ \frac{r^2+1}{r^2-1}\right) | r>1 \} $ $\approx 5.2809$-approximate solution for any instance.\\

\noindent\textbf{Notation.} We introduce some notations that we use throughout this subsection.
Let $\alpha(\lambda,i)$ be the class of hub $i \in H$.
For any $i \in H$, let define $u(\lambda,i)=r^{(\alpha(\lambda,i)-1)+\lambda}$ if $\ell_i \geq 1$, $u(\lambda,i)=0$ if $\ell_i =0$, where $r$ is a real number satisfying $r>1$, i.e.,
\begin{eqnarray}
u(\lambda,i)=\left\{ \begin{array}{ll}
r^{\alpha(\lambda,i)+\lambda-1}& (\ell_i \geq 1), \\
0 & ({\rm \ell_i=0}) .\notag 
\end{array} \right.
\end{eqnarray}
Let define a cost $\hat{c}_{ij}$ for each pair $\{i,j\} \in H^2$ as follows:
\begin{eqnarray}
\hat{c}_{ij}=\left\{ \begin{array}{ll}
|u(\lambda,i)-u(\lambda,j)|& (\alpha(\lambda,i)=\alpha(\lambda,j) \ (\hspace{-0.4cm}\mod 2) ),\\
u(\lambda,i)+u(\lambda,j) & ({\rm otherwise}). \notag
\end{array} \right.
\end{eqnarray}

\noindent\textbf{Remark.} A metric defined by  $\hat{C}=\hat{c}_{ij}$ becomes a line metric (see Figure~\ref{fig:linemetric}).
Thus the matrix $\hat{C}$ is a Monge matrix.

\begin{figure}[htpp]
\begin{center}
			\includegraphics[width=7cm]{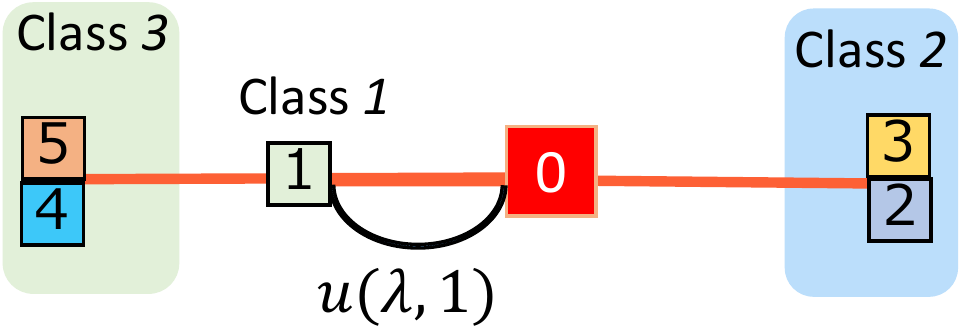}
			\caption{A metric defined by $\hat{C}$}
			\label{fig:linemetric}
\end{center}
\end{figure}

Now we start with the following lemma.
\begin{lemma}
\label{lemma:prob}
Let $\bix^*$ be an optimal solution of LRP.
A vector of random variables $\boldmath{X}$ obtained by the proposed algorithm satisfies that 
${\rm Pr}[\boldmath{X}_{pi}=1]=x_{pi}^* \ (\forall p \in N, \forall i \in H)$.
\end{lemma}

\begin{proof}
\begin{align*}
&{\rm Pr}[\boldmath{X}_{pi}=1]=\Pr[p \in N {\rm \ is\ classified\ into} \ N_{\alpha(\lambda,i)} ]\Pr[p \in N {\rm \ is \ assigned\  to\ } i (\in H)]\\
&=\left(\sum_{j:\alpha(\lambda,j)=\alpha(\lambda,i)}x_{pj}^*\right)\frac{x_{pi}^*/|H^{\alpha(\lambda,i)}|}{\sum_{j:\alpha(\lambda,j)=\alpha(\lambda,i)}x_{pj}^*/|H^{\alpha(\lambda,i)}|}=x_{pi}^*.
\end{align*}
\end{proof}

\begin{lemma}
\label{lemma:k}

For any pair of hubs $\{i,j\} \in H_{\kappa} \times H_{\kappa'}$, 
any real number $r>1$, and any real number $\lambda \in [0,1)$,
we have the inequality $u(\lambda,i)+u(\lambda,j) \leq \frac{r^2+1}{r^2-1} \hat{c}_{ij}$,
where $\kappa,\kappa' \in [\kappa_{\max}]$ and $\kappa \not= \kappa'$. 
\end{lemma}

\begin{proof}

(Case i)
$\kappa -\kappa' \equiv 0 \ (\hspace{-0.2cm}\mod 2)$

\noindent In this case, it is obvious that
\begin{align*}
&u(\lambda,i)+u(\lambda,j)\\
&= \frac{u(\lambda,i)+u(\lambda,j)}{\max \{u(\lambda,i),u(\lambda,j)\}-\min \{u(\lambda,i),u(\lambda,j)\}}\max \{u(\lambda,i),u(\lambda,j)\}-\min \{u(\lambda,i),u(\lambda,j)\} \\
&=\frac{\max \{u(\lambda,i),u(\lambda,j)\}+\min \{u(\lambda,i),u(\lambda,j)\}}{\max \{u(\lambda,i),u(\lambda,j)\}-\min \{u(\lambda,i),u(\lambda,j)\}}\hat{c}_{ij} \\
&=\frac{r^2 \max \{u(\lambda,i),u(\lambda,j)\}+r^2 \min \{u(\lambda,i),u(\lambda,j)\}}{r^2 \max \{u(\lambda,i),u(\lambda,j)\}- r^2 \min \{u(\lambda,i),u(\lambda,j)\}}\hat{c}_{ij}.\\	
\end{align*}
   Recall that $k\not= \kappa'$ and $k-\kappa'\equiv 0 \ (\hspace{-0.2cm} \mod 2)$, and thus it holds that $r^2 {\rm min}\{ u(\lambda,i),u(\lambda,j)\} \leq {\rm max}\{u(\lambda,i),u(\lambda,j)\}$ for any pair of hubs $\{i,j\} \in H_{\kappa} \times H_{\kappa'}$. 
  Then we get
\begin{align*}
&\frac{r^2 \max \{u(\lambda,i),u(\lambda,j)\}+r^2 \min \{u(\lambda,i),u(\lambda,j)\}}{r^2 \max \{u(\lambda,i),u(\lambda,j)\}- r^2 \min \{u(\lambda,i),u(\lambda,j)\}}\hat{c}_{ij}  \\
&\leq
    \frac{r^2 \max \{u(\lambda,i),u(\lambda,j)\}+\max \{u(\lambda,i),u(\lambda,j)\}}{r^2 \max \{u(\lambda,i),u(\lambda,j)\}- \max \{u(\lambda,i),u(\lambda,j)\}}\hat{c}_{ij} 
       =\frac{r^2+1}{r^2-1}\hat{c}_{ij}.
\end{align*}

\noindent(Case ii) $\kappa -\kappa' \equiv 1 \ (\hspace{-0.2cm}\mod 2)$ \\
\noindent From the definition,  we have that

\[
u(\lambda,i)+u(\lambda,j)=\hat{c}_{ij} \leq \frac{r^2+1}{r^2-1}\hat{c}_{ij}.
\]

\end{proof}

\begin{lemma}
\label{lemma:same}
Let $\boldmath{X}$ be a vector of random variables obtained by the proposed algorithm and let  $\bix^*$ be an optimal solution of LRP.
For any pair of non-hubs $(p,q) \in N^2$ 
and any real number $\lambda \in [0,1)$, 
we have the following inequality 
\[
{\rm E}\left[\sum_{\kappa \in [\kappa_{\max}]}\sum_{(i,j)\in H_{\kappa}^2:i\not=j}(\ell_i+\ell_j)\boldmath{X_{pi}}\boldmath{X_{qj}} \right] \leq 2\sum_{i \in H }u(\lambda,i)|x_{pi}^*-x_{qi}^*|.
\]

\end{lemma}

\begin{proof}
First, for any integer $\kappa \in [\kappa_{\max}]$, we show that 
\[
{\rm E}\left[\sum_{(i,j)\in H_{\kappa}^2:i\not=j}(\ell_i+\ell_j)\boldmath{X_{pi}}\boldmath{X_{qj}} \right] \leq 2\sum_{i \in H_{\kappa} }u(\lambda,i)|x_{pi}^*-x_{qi}^*|. \tag{4.1}
\]
(Case i) $\kappa=0$

\noindent We can see that ${\rm E}\left[\sum_{(i,j)\in H_{0}^2:i\not=j}(\ell_i+\ell_j)\boldmath{X_{pi}}\boldmath{X_{qj}} \right] =0 \ (\because \forall (i,j)\in H_0^2, \ell_i=\ell_j=0)$ and $\sum_{i \in H_{\kappa}}u(\lambda,i)|x_{pi}^*-x_{qi}^*|=0 \ (\because \forall i \in H^0, u(\lambda,i)=0)$. Then we obtain the inequality (4.1) for this case.\\

\noindent (Case ii) $\kappa \in \{1,2,\ldots,\kappa_{\max} \}$

\noindent In this case, it is easy to see that
\begin{align}
&{\rm E}\left[\sum_{(i,j)\in H_{\kappa}^2:i\not=j}(\ell_i+\ell_j)\boldmath{X_{pi}}\boldmath{X_{qj}} \right]  =\sum_{(i,j)\in H_{\kappa}^2:i\not=j}\left((\ell_i+\ell_j){\rm Pr}[X_{pi}=X_{qj}=1] \right) 
\notag\\
& \leq \sum_{(i,j)\in H_{\kappa}^2:i\not=j} 2r^{\kappa+\lambda-1}{\rm Pr}[X_{pi}=X_{qj}=1] \ (\because \forall i \in H_{\kappa},  \ell_i \leq 2r^{\kappa+\lambda-1}) \notag\\
&=2r^{\kappa+\lambda-1}\sum_{(i,j)\in H_{\kappa}^2:i\not=j}\Pr[X_{pi}=X_{qj}=1]. \tag{4.2}
\end{align}

\noindent We say that non-hub $p$ and non-hub $q$ are {\bf separated} by a single phase in \textbf{Algorithm~\ref{alg:assign}}
if both $p$ and $q$ are unassigned before the phase and exactly one of $p$ and $q$ is assigned in this phase (See Figure~\ref{fig:somephase}).
Note that even if $p$ and $q$ are separated by some phase,
they may be assigned to a mutual hub later.

\begin{figure}[htpp]
\begin{center}
			\includegraphics[width=7cm]{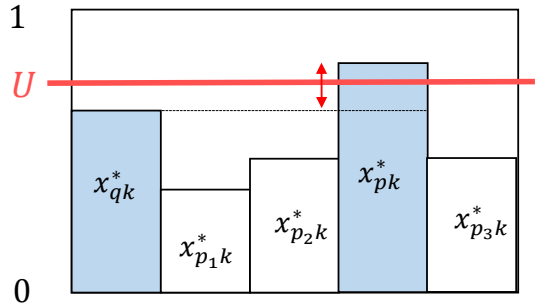}
			\caption{Non-hub $p$ and non-hub $q$ are separated in this phase.}
			\label{fig:somephase}
\end{center}
\end{figure}

The probability $\sum_{(i,j)\in H_{\kappa}^2:i\not=j}{\rm Pr}[X_{pi}=X_{qj}=1]$ in the right-hand side of inequality $(4.2)$ is the probability that both non-hub $p$ and $q$ are classified into $N_{\kappa}$ by \textbf{Algorithm~\ref{alg:class}} and non-hub $p$ and $q$ are assigned to different hubs by \textbf{Algorithm~\ref{alg:assign}}.
This probability can be bounded by the probability that both non-hub $p$ and $q$ are classified into $N_{\kappa}$ by \textbf{Algorithm~\ref{alg:class}} and non-hub $p$ or $q$ are separated by some phase in  \textbf{Algorithm~\ref{alg:assign}}.
Then  for any $\kappa \in \{1,2,\ldots,\kappa_{\max}\}$,
we have that
\begin{align}
&\sum_{(i,j) \in H_{\kappa},i \neq j} \Pr [X_{pi}=X_{qj}=1 ] \notag
=\sum_{(i,j)\in H_{\kappa}^2:i\not=j} {\rm Pr}[\beta(p)=\beta(q)=\kappa]{\rm Pr}[X_{pi}=X_{qj}=1 \bigm| \beta(p)=\beta(q)=\kappa] \\\notag
& \leq
\Pr[\beta(p)=\beta(q)=\kappa] \sum_{k \in H_{\kappa}} \frac {|x_{pk}^*-x_{qk}^*|/|H_{\kappa}|}{{\rm max}\{x_{pk}^*,x_{qk}^*\}/|H_{\kappa}|}. \notag
\end{align}

\noindent Thus, we obtain that
\begin{align}
&2r^{\kappa+\lambda-1}\sum_{(i,j)\in H_{\kappa}^2:i\not=j}\Pr[X_{pi}=X_{qj}=1] \notag\\ 
&\leq
2r^{\kappa+\lambda-1} {\rm Pr}[\beta(p)=\beta(q)=\kappa]\frac{\sum_{i \in H_{\kappa}}|x_{pi}^*-x_{qi}^*|/|H_{\kappa}|}{\sum_{i \in H_{\kappa}}{\rm max}\{x_{pi}^*,x_{qi}^*\}/|H_{\kappa}|} \notag \\
	&\leq 2r^{\kappa+\lambda-1}{\rm Pr}[\beta(p)=k]\frac{\sum_{i \in H_{\kappa}}|x_{pi}^*-x_{qi}^*|}{\sum_{i \in H_{\kappa}}x_{pi}^*} \notag\\
    &=2r^{\kappa+\lambda-1}\sum_{i \in H_{\kappa}}|x_{pi}^*-x_{qi}^*| \ (\because {\rm Pr}[\beta(p)=\kappa]= \sum_{i \in H_{\kappa}}x_{pi}^*)\notag \\
    &=2 \sum_{i \in H_{\kappa}}r^{\kappa+\lambda-1} |x_{pi}^*-x_{qi}^*| =2\sum_{i \in H_{\kappa}}u(\lambda,i)|x_{pi}^*-x_{qi}^*|. \notag
\end{align}
Then we have inequality (4.1) for this case.
From inequality (4.1), we have the desired result: 
\[
{\rm E}\left[\sum_{\kappa \in [\kappa_{\max}]}\sum_{(i,j)\in H_{\kappa}^2:i\not=j}(\ell_i+\ell_j)\boldmath{X_{pi}}\boldmath{X_{qj}} \right] \leq 2\sum_{i \in H }u(\lambda,i)|x_{pi}^*-x_{qi}^*|.
\]
\end{proof}

Next, to show Lemma~\ref{lemma:different}, we first describe Lemma~\ref{lemma:expect=NW} and Theorem~\ref{thm:monge}.
Lemma~\ref{lemma:expect=NW} implies that the probability that non-hub $p$ is classified into $N_{\kappa}$ and non-hub $q$ is classified into $N_{\kappa'}$ by \textbf{Algorithm~\ref{alg:class}} is bounded by $\sum_{i \in H_{\kappa}}\sum_{j \in H_{\kappa'}}y_{piqj}^{NW}$ where $\biy^{NW}$ is a {\em north-west corner rule} solution of the subproblem that is equivalent to a {\em Hitchcock transportation problem (HTP)}. The detail is omitted here (see Appendix).

\begin{lemma}
\label{lemma:expect=NW}
Let $\bm{X}$ be a vector of random variables obtained by the proposed algorithm, let $(\bix,\biy)$ be a feasible solution of LRP and let $\biy^{NW}$ be a solution of HTP {\rm(}defined in Appendix{\rm)}.
obtained by north-west corner rule. For any pair of $\{\kappa,\kappa'\} \in [\kappa_{\max}],\kappa \not = \kappa'$ and any pair of $(p,q)\in N^2
\ (p \not =q)$, we have the following inequality:
\[
\sum_{i \in H_{\kappa}}\sum_{j \in H_{\kappa'}}{\rm E }[X_{pi}X_{qj}]
\leq \sum_{i \in H_{\kappa}}\sum_{j \in H_{\kappa'}}y_{piqj}^{NW}.
\]
\end{lemma}
\noindent The proof is omitted here (see Appendix).

Next we describe well-known relation between a north-west corner rule solution of a Hitchcock transportation problem and the Monge property.

\begin{theorem}\label{thm:monge}
If a given cost matrix $C=(c_{ij})$ is a Monge matrix, then the north-west corner rule solution $\biy^{NW}$ gives an optimal solution of all the Hitchcock transportation problems.
\end{theorem}
\noindent Proof is omitted here (see for example \cite{BEIN1995,BURKARD1996}).
\begin{algorithm}[t]
	\caption{Construct $\biy^*$ from $\bix^*$}\label{alg:x=>y}
    \begin{algorithmic}[1]
   		\REQUIRE{An optimal solution $\bix^*$ of LRP.}
		\ENSURE{ Vectors $\biy^*$} 
        \FOR{$(p,q)\in N^2, p\not=q$}
        	\STATE{Initialize $y_{piqj}^*=0\ (\forall (i,j)\in H^2$)}
            \STATE{Set $y_{piqj}^*$ to $\min \{x_{pi}^*,x_{qi}^*\}  \ (\forall i \in H)$}
            
            \FOR{$i=1,2,3,\ldots,h$}
            	\STATE{$j \leftarrow 1$}
        		\WHILE{$\sum_{k \in H}y_{piqk}^*<x_{pi}^*$}
            	\STATE{Set $y_{piqj}^*$ to $ \min \{x_{qj}^*-\sum_{k \in H}y_{pkqj}^*  , x_{pi}^*-\sum_{k \in H}y_{piqk}^* \}$}
                \STATE{$j \leftarrow j +1$}
            	\ENDWHILE
        	\ENDFOR
        \ENDFOR
      \RETURN{$\biy^*$}
   \end{algorithmic}
\end{algorithm}

Next we consider that we construct a vector $\biy^*$ from the optimal solution $\bix^*$ by \textbf{Algorithm~\ref{alg:x=>y}}.
A vector $\biy^*$  is  optimal to our subproblem HTP.
Note that  we need \textbf{Algorithm~\ref{alg:x=>y}} only for approximation analysis and we don't use it to obtain an approximate solution.
Then we have the following lemma.
\begin{lemma}\label{lemma:x=>y}
Let $\bm{x}^*$ be  an optimal solution of LRP and let $\biy^*$ be a vector obtained by  \textbf{Algorithm~\ref{alg:x=>y}}.
For any pair of $(p,q)\in N^2 \ (p \not =q)$ , we have the following inequality:
\[
\sum_{(i,j)\in H^2:i\not = j}(\ell_i+\ell_j)y_{piqj}^*=\sum_{i \in H}\ell_i|x_{pi}^*-x_{qi}^*|.
\]
\end{lemma}
\noindent Proof is omitted here (see Appendix). 

Now we are ready to prove the following lemma.
\begin{lemma}
\label{lemma:different}
Let $\textbf{X}$ be a vector of random variables obtained by the proposed algorithm.
Let $\bix^*$ be an optimal solution of LRP, and  let $\bm{y}^*$ be  vectors obtained from the optimal solution $\bm{x}^*$ of LRP by \textbf{Algorithm~\ref{alg:x=>y}}.
For any distinct  pair of non-hubs $(p,q) \in N^2 \ (p \not =q)$,
any real number $r>1$,
and any real number $\lambda \in [0,1)$,
we have the following inequality :
\[
{\rm E}\left[\sum_{\{\kappa,\kappa'\} \in [\kappa_{\max}] :\kappa \not=\kappa'}\sum_{i \in H_{\kappa}}\sum_{j \in H_{\kappa'}}(\ell_i+\ell_j)\boldmath{X_{pi}}\boldmath{X_{qj}} \right] \leq \frac{r^2+1}{r^2-1}\sum_{i \in H}u(\lambda,i)|x_{pi}^*-x_{qi}^*|.
\]
\end{lemma}
\begin{proof}
First, we prove the following inequality for any pair of integers  $\{\kappa,\kappa'\} \in [\kappa_{\max}]\ (\kappa \not=\kappa')$ and any  pair of non-hubs $(p,q) \in N^2 \ (p \not=q)$ :
\[
{\rm E}\left[ \sum_{i \in H_{\kappa}}\sum_{j \in H_{\kappa'}}(\ell_i+\ell_j)\boldmath{X_{pi}}\boldmath{X_{qj}} \right] \leq \frac{r^2+1}{r^2-1}\sum_{i \in H_{\kappa}}\sum_{j \in H_{\kappa'}}(u(\lambda,i)+u(\lambda,j)) y_{piqj}^{*}. \tag{5.1} \]

\noindent (Case i) $\kappa,\kappa' \in \{1,2,\ldots,\kappa_{\max} \} \ (\kappa \not= \kappa')$ \\
\noindent In this case, we have the following inequalities from the definition of $u(\lambda,i) \ (i  \in H)$.
\begin{align}
& {\rm E}\left[\sum_{i \in H_{\kappa}}\sum_{j \in H_{\kappa'}}(\ell_i+\ell_j)\boldmath{X_{pi}}\boldmath{X_{qj}} \right] \notag\\
& \leq \sum_{i \in H_{\kappa}}\sum_{j \in H_{\kappa'}}(u(\lambda,i)+u(\lambda,j)) {\rm E}[X_{pi}X_{qj}] \notag\\
& = \sum_{i \in H_{\kappa}}\sum_{j \in H_{\kappa'}}(r^{\kappa+\lambda-1}+r^{\kappa'+\lambda-1}){\rm E}[X_{pi}X_{qj}] \notag \\
& =(r^{\kappa+\lambda-1}+r^{\kappa'+\lambda-1})\sum_{i \in H_{\kappa}}\sum_{j \in H_{\kappa'}}{\rm E}[X_{pi}X_{qj}]. \notag
\end{align}

 
 Using Lemma~\ref{lemma:expect=NW},
 Lemma~\ref{lemma:k},
 and Theorem~\ref{thm:monge},
 we have the following inequalities.
 \begin{align}
& (r^{\kappa+\lambda-1}+r^{\kappa'+\lambda-1})\sum_{i \in H_{\kappa}}\sum_{j \in H_{\kappa'}}{\rm E}[X_{pi}X_{qj}] \\\notag
 &\leq (r^{\kappa+\lambda-1}+r^{\kappa'+\lambda-1})\sum_{i \in H_{\kappa}}\sum_{j \in H_{\kappa'}}y_{piqj}^{NW} \ (\because \rm{ Lemma~\ref{lemma:expect=NW}})\notag \\
  &=\sum_{i \in H_{\kappa}}\sum_{j \in H_{\kappa'}} (u(\lambda,i)+u(\lambda,j)) y_{piqj}^{NW} \notag \\
 & \leq  \frac{r^2+1}{r^2-1}\sum_{i \in H_{\kappa}}\sum_{j \in H_{\kappa'}}\hat{c}_{ij} y_{piqj}^{NW} \ (\because {\rm Lemma~\ref{lemma:k}}) \notag \\
 & \leq  \frac{r^2+1}{r^2-1}\sum_{i \in H_{\kappa}}\sum_{j \in H_{\kappa'}}\hat{c}_{ij} y_{piqj}^{*} \ (\because \hat{C} {\rm \ is\  a\ Monge\ matrix\  and\ Theorem~\ref{thm:monge}}.) \notag\\
& \leq \frac{r^2+1}{r^2-1} \sum_{i \in H_{\kappa}}\sum_{j \in H_{\kappa'}}(u(\lambda,i)+u(\lambda,j))y_{piqj}^{*} \notag
\end{align}
Then we obtained inequality (5.1) for this case.

\noindent(Case ii) $\kappa=0$ or $ \kappa'=0$ \\
We can show  inequality $(5.1)$ for this case by substituting $r^{\kappa+\lambda-1}+r^{\kappa'+\lambda-1}$ in (Case i) by either $r^{\kappa+\lambda-1} $ or $r^{\kappa'+\lambda-1 }$.

Then we obtain that 
\begin{align}
&{\rm E}\left[\sum_{\{\kappa,\kappa'\} \in [\kappa_{\max}] :\kappa \not=\kappa'}\sum_{i \in H_{\kappa}}\sum_{j \in H_{\kappa'}}(\ell_i+\ell_j)\boldmath{X_{pi}}\boldmath{X_{qj}} \right] \notag \\ &\leq \frac{r^2+1}{r^2-1} \sum_{\{\kappa,\kappa'\} \in [\kappa_{\max}]:\kappa \not = \kappa'}\sum_{i \in H_{\kappa}}\sum_{j \in H_{\kappa'}}(u(\lambda,i)+u(\lambda,j))y_{piqj}^{*}  \ (\because {\rm  inequality }\ (5.1)) \notag \\
&\leq \frac{r^2+1}{r^2-1} \sum_{(i,j)\in H^2: i \not =j} (u(\lambda,i)+u(\lambda,j)) y_{piqj}^{*}. \notag\\
& =\frac{r^2+1}{r^2-1} \sum_{i \in H}u(\lambda,i)|x_{pi}^*-x_{qj}^*| . \ (\because {\rm Lemma~\ref{lemma:x=>y}} {\rm \ for \ } u(\lambda,i) {\rm \ instead \ of\ } \ell_i )\notag
\end{align}
\end{proof}

Now, we are ready to show our main theorem.
\begin{theorem}
\label{theorem:main}
The proposed algorithm is $\min \{\frac{r-1}{\log r}\left( 2+ \frac{r^2+1}{r^2-1}\right) | r>1 \}  ( \approx 5.2809 \ {\rm at}\ r \approx 1.91065)$--approximation algorithm for star-star hub-and-spoke network design problems and star-metric labeling problems. 
\end{theorem}

\begin{proof}
Let $\bm{X}$ be a vector of random variables obtained by the proposed algorithm and let $(\bix^*,\biy^*)$ be an optimal solution of LRP.  
For any real number $\lambda \in [0,1)$, we have that
\begin{align}
&{\rm E}[Z] ={\rm E}\left[\sum_{(p,q) \in N^2:p \not =q}w_{pq}\left(\sum_{i \in H}c_{pi}X_{pi}+\sum_{j \in H}c_{qj}X_{qj}+\sum_{(i,j)\in H^2:i\not = j}(\ell_i+\ell_j) X_{pi}X_{qj}\right)\right] \notag \\
&=\sum_{(p,q) \in N^2:p \not =q}w_{pq}\left(\sum_{i \in H}c_{pi}x_{pi}^*+\sum_{j \in H}c_{qj}x_{qj}^*+{\rm E}\left[\sum_{\kappa \in [\kappa_{\max}]}\sum_{(i,j)\in H^2_{\kappa}:i\not = j}(\ell_i+\ell_j)X_{pi}X_{qj}\right]\right. \notag \\
&\hspace{4cm} \left. +{\rm E}\left[\sum_{\{\kappa,\kappa'\} \in [\kappa_{\max}] :\kappa \not=\kappa'}\sum_{i \in H_{\kappa}}\sum_{j \in H_{\kappa'}}(\ell_i+\ell_j)\boldmath{X_{pi}}\boldmath{X_{qj}} \right] \right) (\because {\rm Lemma~\ref{lemma:prob}})\notag \\
& \leq \sum_{(p,q) \in N^2:p \not =q}w_{pq} \left( \sum_{i \in h}c_{pi}x_{pi}^* +\sum_{j \in H}c_{qj}x_{qj}^* +2\sum_{i \in H }u(\lambda,i)|x_{pi}^*-x_{qi}^*|\right. \notag \\
&\hspace{4cm} \left. + \frac{r^2+1}{r^2-1}\sum_{k \in H}u(\lambda,k)|x_{pk}^*-x_{qk}^*| \right)
 \ (\because {\rm Lemma~\ref{lemma:same}} \ {\rm and} \ {\rm Lemma~\ref{lemma:different}})\notag \\
&= \sum_{(p,q) \in N^2:p \not =q}w_{pq} \left( \sum_{i \in H}c_{pi}x_{pi}^* +\sum_{j \in H}c_{qj}x_{qj}^* +\left(2+\frac{r^2+1}{r^2-1} \right)\sum_{k \in H }u(\lambda,k)|x_{pk}^*-x_{qk}^*| \right) \tag{6.1}\end{align}
where $Z$ denotes the objective value of a solution obtained by the proposed algorithm.
Let $\Lambda \in [0,1)$ be a uniform random variable.
The expected value  of $u(\Lambda,k)$ for all $ k \in H$ and for all $ r>1 $ is  ${\rm E}[u(\Lambda,k)] =\int_0^1 r^{\Lambda}\ell_k \ d\Lambda=\frac{r-1}{\log r}\ell_k.$

Thus, from the above discussion and inequality (6.1) which holds for any  $\Lambda \in [0,1)$, we have that
\begin{align}
{\rm E}[Z]& \leq \sum_{(p,q) \in N^2:p \not =q}w_{pq} \left( \sum_{i \in h}c_{pi}x_{pi}^* +\sum_{j \in H}c_{qj}x_{qj}^* +\left(2+\frac{r^2+1}{r^2-1} \right)\sum_{k \in H }{\rm E}[u(\Lambda,k)]|x_{pk}^*-x_{qk}^*| . \right) \notag \\
& =\sum_{(p,q) \in N^2:p \not =q}w_{pq} \left( \sum_{i \in h}c_{pi}x_{pi}^* +\sum_{j \in H}c_{qj}x_{qj}^* +\frac{r-1}{\log r}\left(2+\frac{r^2+1}{r^2-1} \right)\sum_{k\in H }\ell_{k}|x_{pk}^*-x_{qk}^*| . \right) \notag \\
& = \min \left\{\frac{r-1}{\log r}\left( 2+ \frac{r^2+1}{r^2-1}\right) | r>1 \right\} ({\rm optimal \ value \ of \ LRP}) \notag \\
&\leq \min \left\{\frac{r-1}{\log r}\left( 2+ \frac{r^2+1}{r^2-1}\right) | r>1 \right\}   ({\rm optimal \ value \ of \ the \ original \ problem \ SHP}). \notag
\end{align}
Note that when $r>1$, $f(r)=\frac{r-1}{\log r}\left( 2+ \frac{r^2+1}{r^2-1}\right)$ is minimized at $r^* \approx 1.91065$ and we get $f(r^*) \approx 5.2809$. 
Then we obtain the desired result.
\end{proof}


\section{Conclusion}\label{chap:conclusion}
In this paper,
we have studied hub-and-spoke network design problems, motivated by the application to achieve efficient transportation systems.
we considered the case where the set of hubs forms a star, and introduced a star-star hub-and-spoke network design problem and star-metric labeling problem.
The star-metric labeling problem includes the uniform 
labeling problem which is still NP-hard.
 We proposed $\min \left\{\frac{r-1}{\log r}\left( 2+ \frac{r^2+1}{r^2-1}\right) | r>1 \right\}  ( \approx 5.2809 \ {\rm at}\ r \approx 1.91065)$--approximation algorithm for star-star hub-and-spoke network design problems and star-metric labeling problems.
Our algorithms solve a linear relaxation problem and apply dependent rounding procedures.




\section*{Appendix}
\addcontentsline{toc}{section}{Appendix}

\subsection*{Hitchcock Transportation Problems and North-West Corner Rule}
\label{appendix:htp}
A Hitchcock transportation problem is defined 
	on a complete bipartite graph 
	consists of a set of supply points  $A=\{1,2,\ldots ,I\}$ 
	and a set of demand points $B=\{1,2,\ldots , J\}$.
Given a pair of non-negative vectors
	 $(\boldmath{a},\boldmath{b}) \in \mathbb{R}^ I \times \mathbb{R}^J$
	satisfying $\sum_{i =1}^I a_i = \sum_{j=1}^J b_j$ 
	and an $I \times J$ cost matrix $C=(c_{ij})$,
	a Hitchcock transportation problem is formulated as follows: 

\begin{alignat*}{4}
 \mbox{\rm HTP}(\boldmath{a},\boldmath{b}, C):\quad
 & \mbox{\rm min.\quad }
&& 
		\sum_{i=1}^I \sum_{j=1}^J c_{ij}y_{ij}  
	\\
& \mbox{\rm s.~t.}		
&&   \sum_{j=1}^J y_{ij} = a_i  & \quad & (i\in \{1,2,\ldots, I\}), \\ 
&&&  \sum_{i=1}^I y_{ij} = b_j  & \quad & (j\in \{1,2,\ldots, J\}), \\
&&& y_{ij} \geq 0		
	&&(\forall (i,j) \in \{1,2,\ldots,I\} \times \{1,2,\ldots,J\}),
\end{alignat*}
where $y_{ij}$ denotes the amount of flow 
	from a supply point $i \in A$
	to a demand point $j \in B$.

\begin{algorithm}\label{algorithm:northwest}
	\textbf{Algorithm NWCR}
\begin{description}
\item[\textbf{Step 1:}]
	Set all the elements of matrix $Y$ to $0$ and 
	set the target element $y_{ij}$ to $y_{11}$ (top-left corner).
\item[\textbf{Step 2:}]
	Allocate a maximum possible amount of transshipment to the target element \\
	without making the row or column total of the matrix $Y$ exceed 
	the supply or demand respectively.
\item[\textbf{Step 3:}]
If the target element is $y_{IJ}$ (the south-east corner element),
	then stop.
\item[\textbf{Step 4:}]
	Denote the target element by $y_{ij}$.
	If the sum total of $j$th column of $Y$ is equal to $b_j$,
	set the target element to $y_{i j+1}$.
	Else (the sum total of $Y$ of $i$th row is equal to $a_i$),\\
	set the target element to $y_{i+1 j}$.
	Go to Step 2.
\end{description}
\end{algorithm}

We describe north-west corner rule in
	Algorithm~NWCR, which finds
	 a feasible solution of Hitchcock transportation problem
	HTP($\boldmath{a},\boldmath{b},C$).
It is easy to see that the north-west corner rule solution $Y=(y_{ij})$
  satisfies the equalities that 
\[
   \sum_{i=1}^{i'} \sum_{j=1}^{j'} y_{ij}
   = \min \left\{
   \sum_{i=1}^{i'} a_i\;, \;\; \sum_{j=1}^{j'} b_j
          \right\} \;\; 
    ( \forall (i', j') \in \{1,2,\ldots,I\} \times \{1,2,\ldots,J\} ).
\]
 Since the coefficient matrix 
	of the above equality system is nonsingular,
	the north-west corner rule solution is a unique solution 
	of the above equality system.
Thus, the above system of equalities has a unique solution 
	which is feasible to 	HTP($\boldmath{a},\boldmath{b},C$).
    
Next we show that the subproblem of our original problem can be written as a Hitchcock transportation problem.
Let $(\boldmath{x}$, $\boldmath{y})$ be a feasible solution
	of linear relaxation problem.
For any $p \in N$, 
	$\boldmath{x}_p$ denotes a subvector of $\boldmath{x}$
	defined by $(x_{p1}, x_{p2}, \ldots , x_{ph})$.  
When we fix variables  $\boldmath{x}$ in LRP
	to $\boldmath{x}$ and given a pair of $(p,q)\in N^2 
    \ (p \not=q)$, 
	we can decompose the obtained problem into
	Hitchcock transportation problems
	$\{\mbox{HTP}(\boldmath{x}_p, \boldmath{x}_q, \hat{C}(=\hat{c}_{ij}))
		 \mid (p,q) \in N^2\}$
	where 
 \begin{alignat*}{4}
 \mbox{\rm {HTP}$(\boldmath{x}_p, \boldmath{x}_q, \hat{C})$: \quad}
& \mbox{\rm min.\quad }
&& 
		\sum_{i \in H} \sum_{j \in H} \hat{c}_{ij} y_{piqj}  
	\\
& \mbox{\rm s.~t.}		
&&  \sum_{j \in H}  y_{piqj} = x_{pi}  & \quad & (\forall i \in H), \\ 
&&& \sum_{i \in H}  y_{piqj} = x_{qj}  & \quad & (\forall j \in H), \\ 
&&&  y_{piqj} \geq 0			&&(\forall (i,j) \in {H^2}).\\
\end{alignat*}

\subsection*{Monge Property}
We give the definition of a Monge matrix.
A comprehensive research on the Monge property 
	appears in a recent survey~\cite{BURKARD1996}.   
Matrices with this property arise quite often in practical applications, especially in geometric settings.

\begin{definition}\label{definition:monge}
An $m \times n$ matrix $C$ is a Monge matrix
	if and only if $C$ satisfies the so-called Monge property
\[
	c_{ij}+c_{i'j'} \leq c_{ij'}+c_{i'j}\quad\quad\quad
\mbox{\rm for all}\quad1 \leq i <i' \leq m, 1 \leq j < j' \leq n.
\]
\end{definition}
\noindent Note that the north-west corner rule produces an optimal solution of Hitchcock transportation problems
if the cost matrix is a Monge matrix, so we can obtain an optimal solution of HTP$(\boldmath{x}_p, \boldmath{x}_q, \hat{C})$ by north-west corner rule~\cite{BEIN1995}.

\subsection*{Proof of Lemma \ref{lemma:expect=NW}}
\label{sec:expect=NW}
Let $\boldmath{X}$ be a vector of random variables obtained by the proposed algorithm, and let $(\bix,\biy)$ be a feasible solution of LRP. For any pair of $\{\kappa,\kappa'\} \in [\kappa_{\max}] \ (\kappa \not = \kappa')$ and any pair of $(p,q)\in N^2 \ (p \not =q)$, then  we have 
\begin{align}
&\sum_{i \in H_{\kappa}}\sum_{k \in H_{\kappa'}}{\rm E }[X_{pi}X_{qj}]  \leq \left(\sum_{i \in H_{\kappa}}x_{pi}\right) \left(\sum_{j \in H_{\kappa'}}x_{qj}\right) \notag \\
&=\max \{ \sum_{i \in H_{\kappa}}x_{pi},\sum_{j \in H_{\kappa'}}x_{pj}\}\dot \min \{ \sum_{i \in H_{\kappa}}x_{pi},\sum_{j \in H_{\kappa'}}x_{pj}\} \notag \\
& \leq \min \{ \sum_{i \in H_{\kappa}}x_{pi},\sum_{j \in H_{\kappa'}}x_{pj}\} (\because \sum_{i \in H }x_{pi}=1). \tag{6.21}
\end{align}
For any pair of $\{\kappa,\kappa'\} \in [\kappa_{\max}] \ (\kappa \not = \kappa')$ and any pair of $(p,q)\in N^2
\ (p \not =q)$, we have the following Hitchcock transportation problems :

\begin{alignat*}{4}
 \mbox{\rm {HTP}$(\boldmath{x}_p, \boldmath{x}_q, \hat{C})$: \quad}
& \mbox{\rm min.\quad }
&& 
		\sum_{i \in H_{\kappa}} \sum_{j \in H_{\kappa'}} \hat{c}_{ij} y_{piqj}  
	\\
& \mbox{\rm s.~t.}		
&&  \sum_{j \in H_{\kappa'}}   y_{piqj} = x_{pi}  & \quad & (\forall i \in H_{\kappa}), \\ 
&&& \sum_{i \in H_{\kappa}}  y_{piqj} = x_{qj}  & \quad & (\forall j \in H_{\kappa'}), \\ 
&&&  y_{piqj} \geq 0			&&(\forall (i,j) \in {H_{\kappa}} \times H_{\kappa'}).\\
\end{alignat*} 
We see that the north-west corner rule solution $\biy^{NW}=(y^{NW}_{piqj})$
  satisfies the equalities that 
\[
   \sum_{i \in H_{\kappa}}\sum_{j \in H_{\kappa'}} y^{NW}_{piqj}
   = \min \left\{
    \sum_{i \in H_{\kappa}} x_{pi}\;, \;\;\sum_{j \in H_{\kappa'}} x_{qj}
          \right\} \;\; 
    ( \forall  \{\kappa,\kappa'\} \in [\kappa_{\max}],\kappa \not = \kappa' ).
\]
From the equalities and inequality (6.21), we have
\[
\sum_{i \in H_{\kappa}}\sum_{k \in H_{\kappa'}}{\rm E }[X_{pi}X_{qj}] \leq  \sum_{i \in H_{\kappa}}\sum_{j \in H_{\kappa'}} y^{NW}_{piqj}.
\]
Thus, we have the desired result.

\subsection*{Proof of Lemma~\ref{lemma:x=>y}} \label{sec:x=>y}
Let $y^*$ be the vector obtained from an optimal solution of LRP $\bm{x}^*$ by \textbf{Algorithm~\ref{alg:x=>y}}.
Given any distinct pair of non-hubs $(p,q)\in N^2 \ (p\not =q$), 
we can see that  
 \[\sum_{j \in H}y_{piqj}^*=x_{pi}^* \ (\forall i \in H), \]
 \[\sum_{i \in H}y_{piqj}^*=x_{qj}^*\ (\forall j \in H), \]
\[\sum_{j \in H : j \not =i }y_{piqj}= \min \{0, x_{pi}-x_{qi}\} \ (\forall i \in H),\] 
\[\sum_{i \in H : i \not =j }y_{piqj}= \min \{0, x_{qj}-x_{pj}\} \ (\forall j \in H).\]

 Thus we have
\begin{align}
&\sum_{(i,j)\in H^2:i \not = j}(\ell_i+\ell_j)y_{piqj}^*=
\sum_{(i,j)\in H^2:i \not = j}\ell_iy_{piqj}^*+\sum_{(i,j)\in H^2:i \not = j}\ell_jy_{piqj}^* \notag \\
&=\sum_{i \in H}\ell_i \sum_{j \in H: j \not=i }y_{piqj}^*+\sum_{j \in H}\ell_j \sum_{i \in H:j \not = i} y_{piqj}^* \notag \\
&=\sum_{i \in H }\ell_i \min \{0, x_{pi}^*-x_{qi}^*\} +\sum_{j \in H}\ell_j \min \{0, x_{qj}^*-x_{pj}^*\} \notag \\
&= \sum_{i \in H }\ell_i \min \{0, x_{pi}^*-x_{qi}^* \} +\sum_{i \in H} \ell_i \min \{0, x_{qi}^*-x_{pi}^*\} \notag \\
&= \sum_{i \in H} \ell_i |x_{pi}^*-x^*_{qi} |. \notag
\end{align}
Then we have the desired result.
\vspace{2cm}

\begin{figure}[htpp]
\begin{center}
			\includegraphics[width=7cm]{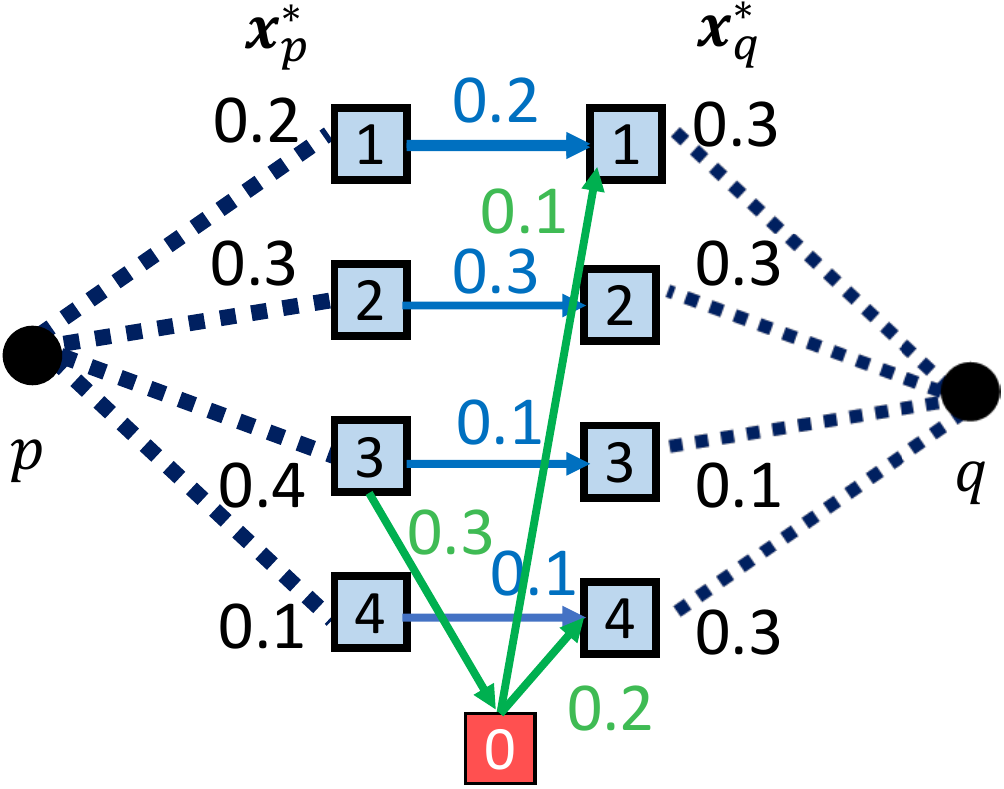}
			\caption{Hitchcock transportation problems with a depot $0$}
			\label{fig:hut_star}
\end{center}
\end{figure}

\setstretch{1.45}
\bibliographystyle{custom-plain}
\bibliography{main-biblio}

\addcontentsline{toc}{section}{Bibliography}

\end{document}